\documentclass[12pt,reqno]{amsart} 

\usepackage{mathdots}
\usepackage{amsmath, amssymb} %% AMS symbol definitions
\usepackage{amscd}            %% for Commutative Diagrams
\usepackage{amsxtra}          %% defines \spcheck, \spdot, etc.
\usepackage{upref}   
\usepackage{enumerate}         %% cross-refs to sections are upright

\usepackage[mathscr]{eucal}

\theoremstyle{plain}

\newtheorem{theorem}{Theorem}

\newtheorem{prop}[theorem]{Proposition}

\newtheorem{theo}[theorem]{Theorem}

\newtheorem{lemma}[theorem]{Lemma}

\newtheorem{definition}[theorem]{Definition}
\newtheorem{example}[theorem]{Example}
\theoremstyle{remark}
\newtheorem{remark}[theorem]{Remark}

\numberwithin{theorem}{section}
\newcommand{\be}%
  {\protect\setcounter{equation}{\value{subsubsection}}}  
\newcommand{\ee}%
\newcommand{\Z}{\mathbb{Z}}

\newcommand{\F}{\mathbb{F}}

\begin{document}

%%%%%%%%%%%%%%%%%%%%%%%%%%%%%%%%%%%%%%%%%%%%%%%%%%%%%%%%%%%%%%%%%%%%%
%%%%%%%%%%%%%%%%%%%%%%%%%%%%%%%%%%%%%%%%%%%%%%%%%%%%%%%%%%%%%%%%%%%%%
%% TOPMATTER
%%%%%%%%%%%%%%%%%%%%%%%%%%%%%%%%%%%%%%%%%%%%%%%%%%%%%%%%%%%%%%%%%%%%%

\title {Negacyclic codes of odd length over the ring $\F_p[u,v]/\langle u^2,v^2,uv-vu\rangle$}
%%%%%

%%%%%
\author{Bappaditya Ghosh} %% This is the correct form!
\let\thefootnote\relax\footnote{Email: bappaditya.ghosh86@gmail.com\\Department of Applied Mathematics,
         Indian School of Mines,
         Dhanbad 826 004,  India.}

\begin{abstract}
We discuss the structure of negacyclic codes of odd length over the ring $\F_p[u, v]/ \langle u^2, v^2, uv-vu \rangle$. We find the unique generating set, the rank and the minimum distance for these negacyclic codes.
\end{abstract}

\maketitle

%%    LEFT AND RIGHT RUNNING HEADS.  MUST COME *AFTER* \maketitle

\markboth{B. Ghosh}{Negacyclic codes of odd length over the ring $R_{u^2,v^2,p}$}
%%%%%%%%%%%%%%%%%%%%%%%%%%%%%%%%%%%%%%%%%%%%%%%%%%%%%%%%%%%%%%%%%%%%%
%% BODY OF PAPER BEGINS HERE
%%%%%%%%%%%%%%%%%%%%%%%%%%%%%%%%%%%%%%%%%%%%%%%%%%%%%%%%%%%%%%%%%%%%%
%%%%%%%%%%%%%%%%%%%%%%%%%%%%%%%%%%%%%%%%%%%%%%%%%%%%%%%%%%%%%%%%%%%%%
%% INTRODUCTION
%%%%%%%%%%%%%%%%%%%%%%%%%%%%%%%%%%%%%%%%%%%%%%%%%%%%%%%%%%%%%%%%%%%%%
\section{Introduction}
The theory of error-correcting codes generally study the codes over the finite field. In recent time, the codes over the finite rings have been studied extensively because of their important role in algebraic coding theory. Negacyclic codes, an important class of constacyclic codes, over finite rings also have been well studied these day.

In 1960's, Berlekamp {\rm\cite{Ber68, Ber84}} introduced negacyclic codes over the field $\F_p$, $p$ odd prime, and designed a decoding algorithm that corrects up to $t\left(<\frac{p-1}{2}\right)$ Lee errors. Wolfmann {\rm\cite{Wolf99}}, in 1999,  studied negacyclic codes of odd length over $\Z_4$. In 2003, Blackford {\rm\cite{Blfd03}} extended these study to negacyclic codes of even length over $\Z_4$. 

The structure of negacyclic codes of length $n$ over a finite chain ring such that the length is not divisible by the characteristic of the residue field is obtained by Dinh and L$\acute{\text{o}}$pez-Permouth {\rm\cite{Dinh-Lopez04}} in a more general setting in the year 2004. When the length $n$ of the code is divisible by the characteristic of the residue field then the code is called a repeated-root codes.  Repeated-root negacyclic codes over finite rings have also been investigated by many authors. The structure of negacyclic codes of length $2^t$ over $\Z_{2^m}$ was obtained in {\rm\cite{Dinh-Lopez04}}. In 2005, Dinh {\rm\cite{Dinh05}} investigated negacyclic codes of length $2^s$ over the Galois ring $\text{GR}(2^a,m)$. S$\breve{\text{a}}$l$\breve{\text{a}}$gean {\rm\cite{Salag06}}, in 2006, has studied the repeated-root negacyclic codes over a finite chain ring and has shown that these codes are principally generated over the Galois ring $\text{GR}(2^a,m)$. Various kinds of distances of negacyclic 
codes of length $2^s$ over $\Z_{2^a}$ are determined in {\rm\cite{Dinhh07}}. The structure of the negacyclic codes of length $2p^s$ over the ring $\F_{p^m}+u\F_{p^m}$ have been discussed in {\rm\cite{Dinh15}}.

Let $p$ be a odd prime. In this paper we study the structure of negacyclic codes of odd length over the non chain ring $\F_p[u,v]/\langle u^2, v^2, uv-vu\rangle$. We find a unique set of generators, rank and a minimal spanning set for these codes. We also find the Hamming distance of these codes for length $p^l$.

The structures of cyclic codes over the ring $R_{u^2,v^2,p} = \F_p[u,v]/\langle u^2, v^2, uv-vu\rangle$ have been discussed in {\rm\cite{Pkk-Bg15}}. We can view the cyclic and negacyclic codes over the ring $R_{u^2,v^2,p}$ as an ideal in the rings $R_{u^2,v^2,p}[x]/\langle x^n-1 \rangle$ and $R_{u^2,v^2,p}[x]/$ $\langle x^n+1 \rangle$ respectively. We define the ring isomorphism from the ring $R_{u^2,v^2,p}[x]/$ $\langle x^n-1 \rangle$ to the ring $R_{u^2,v^2,p}[x]/ \langle x^n+1 \rangle$ to get the structure of negacyclic code over the ring $R_{u^2,v^2,p}$.

\section{Preliminaries}\label{pre}
A linear code $C$ of length $n$ over a ring $R$ is negacyclic if $(-c_{n-1}, c_0, \cdots , c_{n-2}) $ $\in C$ whenever $(c_0, c_1, \cdots , c_{n-1})\in C$. We can consider a negacyclic code $C$ of length $n$ over a ring $R$ as an ideal in the ring $R[x]/\langle x^n+1\rangle$ via the correspondence $R^n\rightarrow R[x]/\langle x^n+1\rangle$, $(c_0, c_1, \cdots , c_{n-1})\rightarrow c_0+c_1x+\cdots+c_{n-1}x^{n-1}$. Let  $R_{u^2,v^2,p}=\F_p + u\F_p+v\F_p + uv\F_p, 
u^2=0$, $v^2=0$ and $uv=vu$. This ring is isomorphic to the ring $\F_p[u,v]/\langle u^2, v^2, uv-vu\rangle$. The ring $R_{u^2,v^2,p}$ is a finite commutative local ring with the unique maximal ideal $\langle u,v\rangle$. The set $\{ \{0\}, \langle u\rangle, \langle v\rangle, \langle uv\rangle, \langle u + \alpha v\rangle$, $\langle u,v\rangle, \langle 1\rangle\}$ gives list of all ideals of the ring $R_{u^2, v^2, p}$, where $\alpha$ is a non zero element of $\F_p$. Since the maximal ideal $\langle u,v\rangle$ is not principal, the ring $R_{u^2, v^2, p}$ is not a chain ring. The residue field  $\overline{R}$ of a ring $R$ is define as $\overline{R} = R/M$, where $M$ is a maximal ideal. For the ring $R_{u^2,v^2,p}$ the residue field is $\F_p$. Let $\mu : R[x] \rightarrow \overline{R}[x]$ denote the natural ring homomorphism that maps $r \mapsto r + M$ and the variable $x$ to $x$.
 We define the degree of the polynomial $f(x) \in R[x]$ as the degree of the polynomial $\mu(f(x))$ in $\overline{R}[x]$, i.e., $deg(f(x)) = deg(\mu(f(x))$ (see, for example, \cite{McDonald74}). A polynomial $f(x) \in R[x]$ is called regular if it is not a zero divisor. The following conditions are equivalent for a finite commutative local ring $R$.
\begin{prop} {\rm (cf. \cite[Exercise XIII.2(c)]{McDonald74})} \label{regular-poly}
Let $R$ be a finite commutative local ring. Let $f(x) = a_0+a_1x+ \cdots +a_nx^n$ be in $R[x]$, then the following are equivalent.
\begin{enumerate}[{\rm (1)}]
 \item $f(x)$ is regular; \label{1}
\item  $\langle a_0, a_1, \cdots , a_n \rangle = R$; \label{4} 
\item  $a_i$ is an unit for some $i$, $0 \leq i \leq n$; \label{3}
\item  $\mu(f(x)) \neq 0$; \label{2}
\end{enumerate}
\end{prop}
Let $g(x)$ be a non zero polynomial in $\F_p[x]$. By above proposition, it is easy to see that the polynomial $g(x)+up_1(x)+vp_2(x)+uvp_3(x) \in R_{u^2,v^2,p}[x]$ is regular. Note that $\text{deg}(g(x)+up_1(x)+vp_2(x)+uvp_3(x)) = \text{deg}(g(x))$.

\section{Structures for negacyclic codes over the ring $R_{u^2, v^2, p}$}\label{generator}
In this section we assume that $n$ is an odd integer. Let  $R_{u^2,v^2,p}=\F_p + u\F_p+v\F_p + uv\F_p, 
u^2=0$, $v^2=0$ and $uv=vu$. The following theorem gives the ring isomorphism from the ring $R_{u^2,v^2,p}[x]/\langle x^n-1 \rangle$ to the ring $R_{u^2,v^2,p}[x]/\langle x^n+1 \rangle$
\begin{prop}\label{isomorphism}
Let $\phi :\frac{R_{u^2,v^2,p}[x]}{\langle x^n-1\rangle}\rightarrow\frac{R_{u^2,v^2,p}[x]}{\langle x^n+1\rangle}$ be a map defined as $\phi(f(x))=f(-x)$, for all $f(x)\in \frac{R_{u^2,v^2,p}[x]}{\langle x^n-1\rangle}$. The map $\phi$ is a ring isomorphism.
\end{prop}
\begin{proof}
For polynomials $f(x), g(x)\in R_{u^2,v^2,p}[x]$,
\[f(x)\equiv g(x) ~\text{mod}~ (x^n-1);\]
if and only if there exists a polynomial $h(x)\in R_{u^2,v^2,p}[x]$ such that
\[f(x)-g(x)=h(x)(x^n-1);\]
if and only if
\[f(-x)-g(-x)=h(-x)((-x)^n-1)=-h(-x)(x^n+1);\]
if and only if
\[f(-x)\equiv g(-x) ~\text{mod}~ (x^n+1);\]
This implies that for $f(x), g(x)\in \frac{R_{u^2,v^2,p}[x]}{\langle x^n-1\rangle}$, $\phi(f(x))=\phi(g(x))$ if and only if $f(x)=g(x)$. Hence, $\phi$ is well-defined and one-to-one. Since the rings  $\frac{R_{u^2,v^2,p}[x]}{\langle x^n-1\rangle}$ and  $\frac{R_{u^2,v^2,p}[x]}{\langle x^n+1\rangle}$ are finite and of same order, $\phi$ is an onto map. It is easy to see that $\phi$ is a ring homomorphism. So $\phi$ is a ring isomorphism.
\end{proof}
\begin{remark} \label{restr-isomorphism}
We restrict the isomorphism $\phi$ to an isomorphism $\phi:\frac{\F_p[x]}{\langle x^n-1\rangle}\rightarrow\frac{\F_p[x]}{\langle x^n+1\rangle}$.
\end{remark}
Throughout this paper we use the isomorphism $\phi$ and restriction of $\phi$ defined in Proposition \ref{isomorphism} and Remark \ref{restr-isomorphism}.
\begin{prop} \label{cyclic-negacyclic}
Let $R$ be a ring. Let $A\subseteq \frac{R[x]}{\langle x^n-1\rangle}$, $B\subseteq \frac{R[x]}{\langle x^n+1\rangle}$ be two sets such that $\phi(A)=B$. Then $A$ is an ideal of $\frac{R[x]}{\langle x^n-1\rangle}$ if and only if $B$ is an ideal of $\frac{R[x]}{\langle x^n+1\rangle}$. Equivalently, $A$ is a cyclic code of length $n$ over the ring $R$ if and only if $B$ is a negacyclic code of length $n$ over the ring $R$.
\end{prop}
\begin{proof}
The proof is obvious since the map $\phi$ is a ring isomorphism.
\end{proof}
\begin{theo} \label{negacyclic}
Let $C$ be a negacyclic code of length $n$ over the ring $R_{u^2,v^2,p}$. Then $C$ will be of the form $C=\langle g_1(x)+ug_{11}(x)+vg_{12}(x)+vug_{13}(x),ug_2(x)+vg_{22}(x)+vug_{23}(x),vg_3(x)+vug_{33}(x),vug_4(x)\rangle$, where $g_4(x)|g_2(x)|g_1(x)|(x^n+1)$ and $g_4(x)|g_3(x)|g_1(x)|(x^n+1)$.
\end{theo}
\begin{proof}
The code $C$ is a negacyclic codes of length $n$ over the ring $R_{u^2,v^2,p}$. From Proposition \ref{cyclic-negacyclic}, we know that for the negacyclic code $C$ there exist a cyclic code, say $A$ over the same ring and of same length. We know the structure of a cyclic code over the ring $R_{u^2,v^2,p}$ from {\rm\cite{Pkk-Bg15}}. Let the cyclic code over the ring $R_{u^2,v^2,p}$ be $A=\langle g(x)+up_1(x)+vq_1(x)+vur_1(x),ua_1(x)+vq_2(x)+vur_2(x),va_2(x)+vur_3(x),vua_3(x)\rangle$, where $a_3(x)|a_1(x)|g(x)|(x^n-1)$ and  $a_3(x)|a_2(x)|g(x)|(x^n-1)$. Now the polynomials $g(x)+up_1(x)+vq_1(x)+vur_1(x)$, $ua_1(x)+vq_2(x)+vur_2(x)$, $va_2(x)+vur_3(x)$, $vua_3(x)\in \frac{R[x]}{\langle x^n-1\rangle}$. Therefore from the definition of $\phi$ from Proposition \ref{isomorphism}, we get $\phi(g(x)+up_1(x)+vq_1(x)+vur_1(x))=g(-x)+up_1(-x)+vq_1(-x)+vur_1(-x)=g_1(x)+ug_{11}(x)+vg_{12}(x)+vug_{13}(x)$, $\phi(ua_1(x)+vq_2(x)+vur_2(x))=ua_1(-x)+vq_2(-x)+vur_2(-x)=ug_2(x)+vg_{22}(x)+vug_{23}(x)$, $\phi(va_2(x)+vur_3(x))
=va_2(-x)+vur_3(-x)=vg_3(x)+vug_{33}(x)$, $\
\phi(vua_3(x))=vua_3(-x)=vug_4(x)$ and $\phi(x^n-1)=-(x^n+1)$, where $g(-x)=g_1(x)$, $a_1(-x)=g_2(x)$, $a_2(-x)=g_3(x)$, $a_3(-x)=g_4(x)$, $p_1(-x)=g_{11}(x)$, $q_1(-x)=g_{12}(x)$, $r_1(-x)=g_{13}(x)$, $q_2(-x)=g_{22}(x)$, $r_2(-x)=g_{23}(x)$, $r_3(-x)=g_{33}(x)$. Again $\phi(A)=C$. Therefore the code $C$ can be written as $C=\langle g_1(x)+ug_{11}(x)+vg_{12}(x)+vug_{13}(x),ug_2(x)+vg_{22}(x)+vug_{23}(x),vg_3(x)+vug_{33}(x), vug_4(x)\rangle$, where $g_4(x)|g_2(x)|g_1(x)|(x^n+1)$ and $g_4(x)|g_3(x)|g_1(x)|(x^n+1)$.
\end{proof}
\begin{theo}\label{unique}
Any negacyclic code $C$ of length $n$ over the ring $R_{u^2,v^2,p}$ is uniquely generated by the polynomials $A_1=g_1(x)+ug_{11}(x)+vg_{12}(x)+uvg_{13}(x), \linebreak A_2=ug_2(x)+vg_{22}(x)+uvg_{23}(x), A_3=vg_3(x)+uvg_{33}(x), A_4=uvg_4(x)$, where, $g_{ij}(x)$ are zero polynomial or $\text{deg}(g_{ij}(x))<\text{deg}(g_{j+1}(x))$ for $1\leq i\leq 3$, $i\leq j\leq 3$.
\end{theo}
\begin{proof}
The code $C$ is generated by the polynomial $A_1, A_2, A_3$ and $A_4$. Let $A=\langle g(x)+up_1(x)+vq_1(x)+vur_1(x),ua_1(x)+vq_2(x)+vur_2(x),va_2(x)+vur_3(x),vua_3(x)\rangle$ be the cyclic code over the ring $R_{u^2,v^2,p}$ such that $\phi(A)=C$. Now, for any polynomial $f(x) \in \F_p[x]$ we get $\text{deg}(f(-x))=\text{deg}(f(x))$. Therefore, $\text{deg}(\phi(f(x)))=\text{deg}(f(x))$. From above theorem, we have $\phi(g(x))=g_1(x)$, $\phi(a_1(x))=g_2(x)$, $\phi(a_2(x))=g_3(x)$, $\phi(a_3(x))=g_4(x)$, $\phi(p_1(x))=g_{11}(x)$, $\phi(q_1(x))=g_{12}(x)$, $\phi(r_1(x))=g_{13}(x)$, $\phi(q_2(x))=g_{22}(x)$, $\phi(r_2(x))=g_{23}(x)$, $\phi(r_3(x))=g_{33}(x)$. Also from Theorem $3.1$ of {\rm\cite{Pkk-Bg15}}, we have $\text{deg}(p_1(x))<\text{deg}(a_1(x))$, $\text{deg}(q_1(x))<\text{deg}(a_2(x))$, $\text{deg}(r_1(x))<\text{deg}(a_3(x))$, $\text{deg}(q_2(x))<\text{deg}(a_2(x))$, $\text{deg}(r_2(x))<\text{deg}(a_3(x))$, $\text{deg}(r_3(x))<\text{deg}(a_3(x))$. Therefore, from the relation $\text{deg}(\phi(f(x)))=\
text{deg}(f(
x))$ we can write 
$\text{deg}(\phi(p_1(x)))<\text{deg}(\phi(a_1(x)))$, $\text{deg}(\phi(q_1(x)))<\text{deg}(\phi(a_2(x)))$, $\text{deg}(\phi(r_1(x)))<\text{deg}(\phi(a_3(x)))$, $\text{deg}(\phi(q_2(x)))<\text{deg}(\phi(a_2(x)))$, $\text{deg}(\phi(r_2(x)))<\text{deg}(\phi(a_3(x)))$, $\text{deg}(\phi(r_3(x)))<\text{deg}(\phi(a_3(x)))$. This implies that $\text{deg}(g_{ij}(x))<\text{deg}(g_{j+1}(x))$ for $1\leq i\leq 3$, $i\leq j\leq 3$.
To prove uniqueness we assume that the polynomial $A_1'=g_1(x)+ug'_{11}(x)+vg'_{12}(x)+uvg'_{13}(x) \in C$ satisfies degree result $\text{deg}(g'_{1j}(x))<\text{deg}(g_{j+1}(x))$ for $1\leq j\leq 3$. Since, $\phi(A)=C$, therefore, there exists a polynomial $g(x)+up'_1(x)+vq'_1(x)+vur'_1(x)\in A$, such that $\phi(g(x)+up'_1(x)+vq'_1(x)+vur'_1(x))=g_1(x)+ug'_{11}(x)+vg'_{12}(x)+uvg'_{13}(x)$, where $\phi(p'_1(x))=p'_1(-x)=g'_{11}(x)$, $\phi(q'_1(x))=q'_1(-x)=g'_{12}(x)$, $\phi(r_1(x))=r'(-x)=g'_{13}(x)$. Since, $\text{deg}(\phi(f(x)))=\text{deg}(f(x))$, for all $f(x) \in A$, thus, $\text{deg}(\phi(p'_1(x)))=\text{deg}(p'_1(x))=\text{deg}(g'_{11}(x))$. Similarly, $\text{deg}(q'_1(x))=\text{deg}(g'_{12}(x))$, $\text{deg}(r'_1(x))=\text{deg}(g'_{13}(x))$. Therefore, $\text{deg}(p'_1(x))=\text{deg}(g'_{11}(x))<\text{deg}(g_2(x))=\text{deg}(a_1(x))$. This implies that $\text{deg}(p'_1(x))<\text{deg}(a_1(x))$ Similarly we get, $\text{deg}(q'_1(x))<\text{deg}(a_2(x))$, $\text{deg}(r'_1(x))<\text{deg}(a_3(x))$. But 
from Theorem $3.1$ of {\rm\cite{Pkk-Bg15}}, we know that the polynomial $g(x)+up_1(x)+vq_1(x)+vur_1(x) \in A$ is unique which satisfying the degree result. Therefore,  $p_1(x)=p'_1(x)$, $q_1(x)=q'_1(x)$ and $r_1(x)=r'_1(x)$. This implies that $\phi(p_1(x))=\phi(p'_1(x))$, thus $g_{11}(x)=g'_{11}(x)$. Similarly, $g_{12}(x)=g'_{12}(x)$ and $g_{13}(x)=g'_{13}(x)$. Hence, $A_1$ is unique. Similarly we can prove that $A_2, A_3$ and $A_4$ are also unique.\\
\end{proof}
\begin{theo}
Let $C=\langle g_1(x)+ug_{11}(x)+vg_{12}(x)+vug_{13}(x),ug_2(x)+vg_{22}(x)+vug_{23}(x),vg_3(x)+vug_{33}(x),vug_4(x)\rangle$ be a negacyclic code of length $n$ over the ring $R_{u^2,v^2,p}$. Then we must have the following properties
\begin{align}
&g_4(x)|g_3(x)|g_1(x), g_4(x)|g_2(x)|g_1(x)|(x^n+1),\\
&g_{i+1}(x)|\frac{x^n+1}{g_i(x)}g_{ii}(x), ~ {\rm{for}} ~ 1 \leq i \leq 3,\\
&g_3(x)|\frac{g_1(x)}{g_2(x)}g_{22}(x)\\
&g_4(x)|g_{22}(x)\\
&g_4(x)|\left(g_{11}(x)-\frac{g_1(x)}{g_3(x)}g_{33}(x)\right)\\
&g_4(x)|\left(g_{12}(x)-\frac{g_1(x)}{g_2(x)}g_{23}(x)+\frac{g_1(x)}{g_2(x)g_3(x)}g_{22}(x)g_{33}(x)\right)\\
&g_{i+j+1}(x)|\frac{x^n+1}{g_i(x)}s_{i(i+j)} ~{\rm{for}}~ 1 \leq i \leq 2 ~{\rm{and ~ for ~ a ~ fix}}~ i  ~{\rm{for}}~  1 \leq j \leq 3-i,\notag \\&{\rm{where}},~ s_{ii}=g_{ii} ~ {\rm{and}} ~ s_{i(i+j)}=g_{i(i+j)}-\sum_{l=1}^j\frac{s_{i(i+l-1)}}{g_{i+l}(x)}g_{(i+l)(i+j)}(x).
\end{align}
\end{theo}
\begin{proof}
Let $A=\langle g(x)+up_1(x)+vq_1(x)+vur_1(x),ua_1(x)+vq_2(x)+vur_2(x),\linebreak va_2(x)+vur_3(x),vua_3(x)\rangle$ be the cyclic code over the ring $R_{u^2,v^2,p}$ such that $\phi(A)=C$. Also, from Theorem \ref{negacyclic}, we have $\phi(g(x))=g_1(x)$, $\phi(a_1(x))=g_2(x)$, $\phi(a_2(x))=g_3(x)$, $\phi(a_3(x))=g_4(x)$, $\phi(p_1(x))=g_{11}(x)$, $\phi(q_1(x))=g_{12}(x)$, $\phi(r_1(x))=g_{13}(x)$, $\phi(q_2(x))=g_{22}(x)$, $\phi(r_2(x))=g_{23}(x)$, $\phi(r_3(x))=g_{33}(x)$. Now, from Remark \ref{restr-isomorphism}, we get that the map $\phi$, $\phi:\frac{\F_p[x]}{\langle x^n-1\rangle}\rightarrow\frac{\F_p[x]}{\langle x^n+1\rangle}$ such that $\phi(f(x))=f(-x)$, $\forall ~ f(x)\in \frac{\F_p[x]}{\langle x^n-1\rangle}$ is an isomorphism. Now from Proposition $3.2$ of {\rm\cite{Pkk-Bg15}}, we know that the properties are true for the ring $\frac{\F_p[x]}{\langle x^n-1\rangle}$. Therefore all of these properties are true for the ring $\frac{\F_p[x]}{\langle x^n+1\rangle}$.
\end{proof}
The following theorem characterizes the free negacyclic codes over the ring  $R_{u^2,v^2,p}$.
\begin{theo} \label{freecode}
If $C=\langle g_1(x)+ug_{11}(x)+vg_{12}(x)+vug_{13}(x),ug_2(x)+vg_{22}(x)+vug_{23}(x),vg_3(x)+vug_{33}(x),vug_4(x)\rangle$ be a negacyclic code of length $n$ over the ring $R_{u^2,v^2,p}$, then $C$ is a free negacyclic code if and only if $g_1(x)=g_4(x)$. In this case, we have $C=\langle g_1(x)+ug_{11}(x)+vg_{12}(x)+vug_{13}(x)\rangle$ and $g_1(x)+ug_{11}(x)+vg_{12}(x)+vug_{13}(x)|(x^n+1)$ in $R_{u^2,v^2,p}$.
\end{theo}
\begin{proof}
We are given that $C$ is a negacyclic code over the ring $R_{u^2,v^2,p}$. Hence from Proposition \ref{cyclic-negacyclic}, there exist one and only one cyclic code $A$ over the ring $R_{u^2,v^2,p}$ such that $\phi(A)=C$. Let the cyclic code be $A=\langle g(x)+up_1(x)+vq_1(x)+vur_1(x),ua_1(x)+vq_2(x)+vur_2(x),va_2(x)+vur_3(x),vua_3(x)\rangle$, where\linebreak $a_3(x)|a_1(x)|g(x)|(x^n-1)$ and $a_3(x)|a_2(x)|g(x)|(x^n-1)$. Therefore we have, $\phi(g(x)+up_1(x)+vq_1(x)+vur_1(x))=g_1(x)+ug_{11}(x)+vg_{12}(x)+vug_{13}(x)$, $\phi(ua_1(x)+vq_2(x)+vur_2(x))=ug_2(x)+vg_{22}(x)+vug_{23}(x)$, $\phi(va_2(x)+vur_3(x))=vg_3(x)+vug_{33}(x)$, $\phi(vua_3(x))=vug_4(x)$ and $\phi(x^n-1)=-(x^n+1)$. Now, it is given $g_1(x)=g_4(x)$. Since $\phi$ is an isomorphism therefore $g(x)=a_3(x)$. We know from Proposition $3.3$ of {\rm\cite{Pkk-Bg15}} that $A=\langle g(x)+up_1(x)+vq_1(x)+vur_1(x)\rangle$ if and only if $g(x)=a_3(x)$. Now $\phi(g(x)+up_1(x)+vq_1(x)+vur_1(x))=g_1(x)+ug_{11}(x)+vg_{12}(x)+vug_{13}(x)$, Hence $C=\langle g_1(x)
+ug_{11}(x)+vg_{12}(x)+vug_{13}(
x)\rangle$. Again we have $\phi(x^n-1)=-(x^n+1)$ and we know from Proposition $3.3$ of {\rm\cite{Pkk-Bg15}} that $g(x)+up_1(x)+vq_1(x)+vur_1(x)|(x^n-1)$. Hence $g_1(x)+ug_{11}(x)+vg_{12}(x)+vug_{13}(x)|(x^n+1)$.
\end{proof}
Note that we get the simpler form for the generators of the negacyclic code over $R_{u^2,v^2,p}$, like in the above theorem, if we have $g_1(x) = g_2(x), g_3(x)$ or $g_4(x) = g_2(x), g_3(x)$.\\

In the following theorem we write the structure of $C$ when $n$ be relatively prime to $p$.

\begin{theo}
Let $C$ be a negacyclic code over the ring $R_{u^2,v^2,p}$ of length $n$. If $n$ is relatively prime to $p$, then we have $C=\langle g_1(x)+ug_2(x)+uvg_{13}(x), vg_3(x)+uvg_4(x)\rangle$ with $g_2(x)|g_1(x)|(x^n+1)$ and $ g_4(x)|g_3(x)|g_1(x)|(x^n+1)$.
\end{theo}
\begin{proof}
Let $C$ be a negacyclic code over the ring $R_{u^2,v^2,p}$. Hence from Proposition \ref{cyclic-negacyclic}, there exists one and only one cyclic code $A$ over the ring $R_{u^2,v^2,p}$ such that $\phi(A)=C$. If $n$ is relatively prime to $p$, then from Theorem $3.4$ of {\rm\cite{Pkk-Bg15}}, we can write the cyclic code $A=\langle g(x)+ua_1(x)+uvr_1(x), va_2(x)+uva_3(x)\rangle$ with $a_1(x)|g(x)|(x^n-1)$ and $ a_3(x)|a_2(x)|g(x)|(x^n-1)$.
Let $\phi(g(x)+ua_1(x)+uvr_1(x))=g_1(x)+ug_2(x)+uvg_{13}(x)$ and $\phi(va_2(x)+uva_3(x))=vg_3(x)+uvg_4(x)$, where, $g(-x)
=g_1(x)$, $r_1(-x)=g_{13}(x)$, $a_1(-x)=g_2(x)$, $a_2(-x)=g_3(x)$, $a_3(-x)=g_4(x)$. Therefore $C$ can be written as $C=\langle g_1(x)+ug_2(x)+uvg_{13}(x), vg_3(x)+uvg_4(x)\rangle$ with $g_2(x)|g_1(x)|(x^n+1)$ and $ g_4(x)|g_3(x)|g_1(x)|(x^n+1)$.
\end{proof}
\section{The Ranks and the minimum distance}\label{rank}
In this section, we find the rank and minimal spanning set of negacyclic codes over the ring $R_{u^2,v^2,p}$. Following Dougherty and Shiromoto \cite[page 401]{Dou-Shiro01}, we define the rank of the code $C$ by the minimum number of generators of $C$ and define the free rank of $C$ by the maximum of the ranks of $R_{u^2,v^2,p}$-free submodules of $C$.\\
\begin{theo}
Let $n$ be not relatively prime to $p$. Let $C$ be a negacyclic code over the ring $R_{u^2,v^2,p}$ of length $n$. If $C=\langle g_1(x)+ug_{11}(x)+vg_{12}(x)+vug_{13}(x),ug_2(x)+vg_{22}(x)+vug_{23}(x),vg_3(x)+vug_{33}(x),vug_4(x)\rangle$ with $deg(g_1(x))\linebreak =r_1$, $deg(g_2(x))=r_2$, $deg(g_3(x))=r_3$, $deg(g_4(x))=r_4$, then the minimal spanning set of $C$ is $B=\{A_1, xA_1, \cdots , x^{n-r_1-1}A_1, A_2, xA_2, \cdots , x^{r_1-r_2-1}A_2,\linebreak A_3, xA_3, \cdots , x^{r_1-r_3-1}A_3, A_4, xA_4, \cdots, x^{r'-r_4-1}A_4\}$, where, $r'=min\{r_2, r_3\}$ and $A_1=g_1(x)+ug_{11}(x)+vg_{12}(x)+vug_{13}(x)$, $A_2= ug_2(x)+vg_{22}(x)+vug_{23}(x)$, $A_3=vg_3(x)+vug_{33}(x)$, $A_4=vug_4(x)$ also $C$ has free rank $n-r_1$ and rank $n+r_1+r'-r_2-r_3-r_4$.
\end{theo}
\begin{proof}
Let $C$ be a negacyclic code over the ring $R_{u^2,v^2,p}$ of length $n$, where $n$ is not relatively prime to $p$. From the Proposition \ref{cyclic-negacyclic}, we get that there exists a cyclic code $A$ over the same ring such that $\phi(A)=C$. Now, we know from Theorem $4.1$ of {\rm\cite{Pkk-Bg15}}, the minimal spanning set of the cyclic code $A$ over the ring $R_{u^2,v^2,p}$ is $\{g(x)+up_1(x)+vq_1(x)+vur_1(x), x(g(x)+up_1(x)+vq_1(x)+vur_1(x)), \cdots, x^{n-r_1-1}(g(x)+up_1(x)+vq_1(x)+vur_1(x)), ua_1(x)+vq_2(x)+vur_2(x), x(ua_1(x)+vq_2(x)+vur_2(x)), \cdots, x^{r_1-r_2-1}(ua_1(x)+vq_2(x)+vur_2(x)), va_2(x)+vur_3(x), x(va_2(x)+vur_3(x)), \cdots, x^{r_1-r_3-1}(va_2(x)+vur_3(x)),\linebreak vua_3(x), x(vua_3(x)), \cdots, x^{r'-r_4-1}(vua_3(x))\}$, where, $r'=min\{r_2, r_3\}$. From Theorem \ref{negacyclic} we have $\phi(g(x)+up_1(x)+vq_1(x)+vur_1(x))=A_1, \phi(ua_1(x)+vq_2(x)+vur_2(x))=A_2, \phi(va_2(x)+vur_3(x))=A_3$ and $\phi(vua_3(x))=A_4$. Therefore the spanning set of negacyclic code $C$ over the ring $R_{
u^2,v^2,p}$ is $B=\{A_1, xA_1, \cdots , x^{n-r_
1-1}A_1, A_2, xA_2, \cdots , x^{r_1-r_2-1}A_2, A_3, xA_3, \cdots , x^{r_1-r_3-1}A_3, A_4,\linebreak xA_4, \cdots, x^{r'-r_4-1}A_4\}$, where, $r'=min\{r_2, r_3\}$ and $A_1=g_1(x)+ug_{11}(x)+vg_{12}(x)+vug_{13}(x)$, $A_2= ug_2(x)+vg_{22}(x)+vug_{23}(x)$, $A_3=vg_3(x)+vug_{33}(x)$, $A_4=vug_4(x)$. 
\end{proof}

Let $n$ be a positive integer not relatively prime to $p$. Let $C$ be a negacyclic code of length $n$ over the ring $R_{u^2,v^2,p}$. We know that there exists a cyclic code $A$ of length $n$ over the ring $R_{u^2,v^2,p}$ such that $\phi(A)=C$, where, $\phi$ is defined as $\phi(f(x))=f(-x)$, for  $f(x) \in A$. The following lemma shows that the isomorphism $\phi$ is a distance preserving map.\\
\begin{lemma}\label{lm}
Let $\phi(A)=C$, where, $A$ and $C$ are the cyclic and negacyclic codes of length $n$ over the ring $R_{u^2,v^2,p}$ and $\phi$ is defined as $\phi(f(x))=f(-x)$, for  $f(x) \in A$, then, $w_{H}(f(x)) = w_{H}(\phi(f(x)))$.
\end{lemma}
\begin{proof}
Let $f(x)=\sum_{i=0}^{n-1}c_ix^i$, where, $c_i \in R_{u^2,v^2,p}$. Now $\phi(f(x))=f(-x)=\sum_{i=0}^{n-1}(-1)^ic_ix^i$. Therefore the coefficient of $x^i$ of $f(x)$ and $f(-x)$ are $c_i$ and $(-1)^ic_i$. That is both coefficient are simultaneously $0$ or non $0$. Hence, $w_{H}(f(x))=w_{H}(\phi(f(x)))$.
\end{proof}
\begin{theorem} \label{md1}
Let $n$ be not relatively prime to $p$. If $ C = \langle A_1, A_2, A_3, A_4\rangle$ is a negacyclic code of length $n$ over the ring $R_{u^2,v^2,p}$. Then $w_{H}(C)=w_{H}(A)$, where, $A$ is the cyclic codes over the ring $R_{u^2,v^2,p}$ such that $\phi(A)=C$.
\end{theorem}
\begin{proof}
Let $h(x)$ be the minimum weighted polynomial in $A$ and the weight is $w_{H}(h(x))=m$. There exists a polynomial $f(x) \in C$ such that $\phi(h(x))=f(x)$. From Lemma \ref{lm}, the weight of $w_{H}(f(x))=m$. Now we prove that $f(x)$ is the minimum weighted polynomial in $C$. If possible, let $f_1(x)$ be the minimum weighted polynomial of $C$ and $w_{H}(f_1(x)) < m$. There exists a polynomial $h_1(x) \in A$ such that $\phi(h_1(x))=f_1(x)$. Again, from Lemma \ref{lm}, $w_{H}(h_1(x))=w_{H}(f_1(x))<m$. Hence, a contradiction that $h(x)$ be the minimum weighted polynomial in $A$. Therefore, $w_{H}(C)=w_{H}(A)$.
\end{proof}
\begin{definition}
Let $ m = b_{l-1}p^{l-1} + b_{l-2}p^{l-2} + \cdots + b_1p + b_0$, $b_i \in \F_p, 0 
\leq i \leq l-1$, be the $p$-adic expansion of $m$.
\begin{enumerate} [{\rm (1)}]
 \item If $ b_{l-i}  \neq 0$ for all $1  \leq i \leq q, q < l, $ and $ b_{l-i} = 0 $ for all $i, q+1 \leq i \leq l$, then $m$ is said to have a $p$-adic length $q$ zero expansion.
\item If $ b_{l-i}  \neq 0$ for all $1  \leq i \leq q, q < l, $ $b_{l-q-1} = 0$ and $ b_{l-i} \neq 0 $ for some $i, q+2 \leq i \leq l$, then $m$ is said to have  $p$-adic length $q$ non-zero expansion.
\item If $ b_{l-i}  \neq 0$ for $1  \leq i \leq l, $ then $m$ is said to have a $p$-adic length $l$  expansion or $p$-adic full expansion.
\end{enumerate}
\end{definition}
The following theorem follows from the above theorem and Theorem 5.4 of \cite{Pkk-Bg15}.
\begin{theorem} \label{md-thm}
Let $C$ be a negacyclic code over the ring $R_{u^2,v^2,p}$ of length $p^l$ where $l$ is a positive integer. Then,  $C=\langle A_1, A_2, A_3, A_4 \rangle$, where, $g_1(x) = (x+1)^{t_1}, g_2(x) = (x+1)^{t_2}, g_3(x) = (x+1)^{t_3}, g_4(x) = (x+1)^{t_4}$,  for some $t_1 > t_2 > t_4 > 0$, $t_1 > t_3 > t_4 > 0$ $(where A_i$'s and $g_i$'s are defined in Theorem \ref{unique}$)$
\begin{enumerate}[{\rm (1)}]
\item If $t_4 \leq p^{l-1},$ then $d(C) = 2$. 
\item If $t_4 > p^{l-1}$, let $t_4 = b_{l-1}p^{l-1} + b_{l-2}p^{l-2} + \cdots + b_1p + b_0$ be the $p$-adic expansion of $t_4$ and $g_4(x)=(x+1)^{t_4}=(x^{p^{l-1}}+1)^{b_{l-1}}(x^{p^{l-2}}+1)^{b_{l-2}} \cdots (x^{p^{1}}+1)^{b_1}(x^{p^0}+1)^{b_0}$.
\begin{enumerate}[{\rm ($a$)}]
 \item If $t_4$ has a $p$-adic length $q$ zero expansion or full expansion $(l=q)$, then $d(C) = (b_{l-1}+1)(b_{l-2}+1)\cdots(b_{l-q}+1)$.
\item If $t_4$ has a $p$-adic length $q$ non-zero expansion, then $d(C)=2(b_{l-1}+1)(b_{l-2}+1)\cdots(b_{l-q}+1)$.
\end{enumerate}
\end{enumerate}
\end{theorem}
%\begin{proof}
%Let $A$ be a cyclic codes over the ring $R_{u^2,v^2,p}$ such that  $\phi(A)=C$. Let $A=\langle g(x)+up_1(x)+vq_1(x)+vur_1(x),ua_1(x)+vq_2(x)+vur_2(x),va_2(x)+vur_3(x),vua_3(x)\rangle$ and $C=\langle g_1(x)+ug_{11}(x)+vg_{12}(x)+vug_{13}(x),ug_2(x)+vg_{22}(x)+vug_{23}(x),vg_3(x)+vug_{33}(x),vug_4(x)\rangle$, where, $g_1(x) = (x+1)^{t_1}, g_2(x) = (x+1)^{t_2}, g_3(x) = (x+1)^{t_3}, g_4(x) = (x+1)^{t_4}$. Therefore, $g(x)=(-1)^{t_1}(x-1)^{t_1}, a_1(x)=(-1)^{t_2}(x-1)^{t_2}, a_2(x)=(-1)^{t_3}(x-1)^{t_3}, a_3(x)=(-1)^{t_4}(x-1)^{t_4}$.\\
%$(1)$ From the Theorem [~] of [paper], we know that, if $t_4 \leq p^{l-1},$ then $d(A)=2$. Therefore from Lemma 5.2, we have $d(C)=2$.\\
%$(2)$ Let $t_4 > p^{l-1}$, and $t_4 = b_{l-1}p^{l-1} + b_{l-2}p^{l-2} + \cdots + b_1p + b_0$ be the $p$-adic expansion of $t_4$.\\
%$(a)$ If $t_4$ has a $p$-adic length $q$ zero expansion or full expansion $(l=q)$, then from the Theorem [~] of [paper], we know that, $d(A) = (b_{l-1}+1)(b_{l-2}+1)\cdots(b_{l-q}+1)$. Therefore from Lemma 5.2, we have $d(C)=(b_{l-1}+1)(b_{l-2}+1)\cdots(b_{l-q}+1)$.\\
%$(b)$ Again, from the Theorem [~] of [paper], we know that, If $t_4$ has a $p$-adic length $q$ non-zero expansion, then $d(A)=2(b_{l-1}+1)(b_{l-2}+1)\cdots(b_{l-q}+1)$. Hence, from Lemma 5.2, we have $d(C)=2(b_{l-1}+1)(b_{l-2}+1)\cdots(b_{l-q}+1)$.
%\end{proof}
\section{Examples} \label{exm}

\begin{example}
Negacyclic codes of length $5$ over $R_{u^2,v^2,5} = \F_5 + u \F_5 + v \F_5 + uv \F_5, u^2 = 0, v^2 = 0, uv = vu$: We have
$$ x^5+1=(x+1)^5 ~\text{over}~ R_{u^2,v^2,5}.$$
Let $g=x+1$ and $c_0, c_1, c_2, c_3, c_4, c_5 \in \F_5$. The non-zero negacyclic codes of length $5$ over $R_{u^2,v^2,5}$ with generator polynomial, rank and minimum distance are given in tables below.\\
\end{example}

\begin{center}
{\bf Table 1.} All non zero free negacyclic codes of length 5 over $R_{u^2,v^2,5}$.
\begin{tabular}{| l | c| c |}
\hline
Non-zero generator polynomials & Rank & d(C)\\
\hline
$<g^4+uc_0g^3+vc_1g^3+uvc_2g^3>$, $c_0c_1=0$&1&5\\
\hline
$<g^3+uc_0g^2+vc_1g^2+uv(c_2+c_3x)g>$&2&4\\
\hline
$<g^2+u(c_0+c_1x)+v(c_2+c_3x)+uv(c_4+c_5x)>$, &3&3\\ $c_0=c_1$ or $c_2=c_3$&&\\
\hline
$<g+uc_0+vc_1+uvc_2>$&4&2\\
\hline
$<1>$ &5&1\\
\hline
\end{tabular}\\
\end{center}
\newpage
{\bf Table 2.} All non zero non free single generated negacyclic codes of length 5 over $R_{u^2,v^2,5}$.
\begin{center}
\begin{tabular}{| l | c| c |}
\hline
Non-zero generator polynomials & Rank & d(C)\\
\hline
$<ug^4+vc_0g^4+uvc_1g^3>$&1&5 \\
\hline
$<vg^4+uvc_0g^3>$&1&5 \\
\hline
$<uvg^4>$&1&5 \\
\hline
$<ug^3+v(c_0+c_1x)g^3+uv(c_2+c_3x)g>$&2&4 \\
\hline
$<vg^3+uv(c_0+c_1x)g>$&2&4\\
\hline
$<uvg^3>$&2&4\\
\hline
$<ug^2+v(c_0+c_1x+c_2x^2)g^2+uv(c_3+c_4x)>$&3&3\\
\hline
$<vg^2+uv(c_0+c_1x)>$&3&3\\
\hline
$<uvg^2>$&3&3\\
\hline
$<ug+v(c_0+c_1x+c_2x^2+c_3x^3)g+uvc_4>$&4&2\\
\hline
$<vg+uvc_0>$&4&2\\
\hline
$<uvg>$&4&2\\
\hline
$<u+v(c_0+c_1x+c_2x^2+c_3x^3+c_4x^4)>$&5&1 \\
\hline
$<v>$&5&1 \\
\hline
$<uv>$&5&1 \\
\hline
\end{tabular}\\
\end{center}

\begin{center}
{\bf Table 3.} Some non zero non free negacyclic codes of length 5 over $R_{u^2,v^2,5}$.

\begin{tabular}{| l | c| c |}
\hline
Non-zero generator polynomials & Rank & d(C)\\
\hline
$<g^4+uc_0g^3+vc_1g^3+uvc_2g^2, uvg^3>$&2&4\\
\hline
$<ug^4+uvc_0g^3, vg^4+uvc_1g^3>$&2&5\\
\hline
$<ug^4+v(c_0+c_1x)g^3+uvg^2, uvg^3>$&2&4\\
\hline
$<ug^4+vc_0g^3+uvc_1g^2, vg^4>$&2&5\\
\hline
$<ug^4+uvc_0g^2, uvg^3>$&2&4\\
\hline
$<vg^4+uvc_0g^2, uvg^3>$&2&4\\
\hline
$<vg^4+uvc_0g, uvg^2>$ &3&3\\
\hline
$<g^3+uc_0g+vc_1g+uvc_2, ug^2+vc_3g+uvc_4$,&5&2\\ $vg^2+uvc_5, uvg>$, $c_0c_2=0$&&\\
\hline
$<ug^3+v(c_0+c_1x)g^3+uv(c_2+c_3x), uvg^2>$&3&3\\
\hline
$<vg^3+uvc_0, uvg>$ &4&2\\
\hline
$<g^2+uc_0+vc_1, ug+vc_2, vg, uv>$&6&1\\
\hline
$<ug^2+vc_0+uvc_1, vg^2+uvc_2, uvg>$&7&2\\
\hline
$<vg^2+uvc_2, uvg>$&4&2\\
\hline
$<g+uc_0+vc_1, uv>$&5&1\\
\hline
$<g+uc_0,  v>$&5&1\\
\hline
$<g+vc_0, u+vc_1 >$&5&1\\
\hline
$<g, u, v>$&6&1\\
\hline
$<ug+vc_0, vg, uv>$&9&1 \\
\hline
$<vg, uv>$&5&1 \\
\hline
$< u, v >$&10&1 \\
\hline
\end{tabular}
\end{center}

\bibliographystyle{plain}
\bibliography{ref}
\end{document}